\documentclass{article}

\usepackage{hyperref}

\usepackage{color}
\usepackage{amsmath}
\usepackage{amsfonts}
\usepackage{amssymb}
\usepackage{theorem}
\usepackage[dvips]{graphicx}

\theorembodyfont{\it}
\newtheorem{theorem}{Theorem}[section]
\newtheorem{corollary}[theorem]{Corollary}

\newtheorem{lemma}[theorem]{Lemma}
\newtheorem{proposition}[theorem]{Proposition}

\theorembodyfont{\rm}
\newtheorem{definition}[theorem]{Definition}

\newtheorem{remark}[theorem]{Remark}

\newtheorem{example}[theorem]{Example}

\numberwithin{equation}{section}

\newenvironment{proof}{\noindent{\it Proof. }}{\hfill$\square$\medskip}

\def \cH {{\cal H}}
\def \R  {\mathbb{R}}
\def \C  {\mathbb{C}}
\def \T  {\mathbb{T}}
\def \Z  {\mathbb{Z}}
\def \N  {\mathbb{N}}
\def \Q  {\mathbb{Q}}
\def \Re {{\rm Re}}
\def \Im {{\rm Im}}

\def \ot {\otimes}
\def \od {\odot}
\def \e  {\epsilon}

\def \la {\lambda}

\def \t  {\tilde}

\def \co  {{\rm co\,}}

\def \tr  {{\rm tr}}

\def \span{{\rm span\,}}

\def \Hom {{\rm Hom\,}}
\def \dim {{\rm dim\,}}
\def \codim {{\rm codim\,}}

\newcommand{\skladnik}[2]{\Big(\conj{(3 #1_j - #2_j)}(3 #2_l - #1_l)+\frac{\delta_{jl}}{\alpha}\Big) e^{-{\alpha}|#1- #2|^2}}

\def \cC {{\cal C}}
\newcommand{\conj}[1]{{#1}^*}

\def \cpar {\conj{\partial}}

\def \oh {\odot}
\def \e {\epsilon}
\def \f {\gamma}

\newcommand{\tphi}{v}
\newcommand{\tpsi}{w}
\newcommand{\tf}{\tilde{\gamma}}
\newcommand{\tie}{\tilde{\epsilon}}
\renewcommand{\Upsilon}{F}

\newcommand{\cphi}{[\tphi]}

\newcommand{\q}[1]{\overline{#1}}
\newcommand{\Ker}{\ker_K\!\rho}
\newcommand{\keR}{\ker_L\!\rho}

\newcommand{\pcqd}[2]{{(#1,#2)}}
\newcommand{\pd}{ \pcqd{p}{d} }
\newcommand{\pc}{ \pcqd{p}{c} }
\newcommand{\qd}{ \pcqd{q}{d} }

\title{{\bf Integral representations of separable states
\footnote{Research supported by the Ministry of Research and
Higher Education, grant N201 039 32/2703, 2007-2010. E-mail:
b.jakubczyk@impan.gov.pl, g.pietrzkowski@impan.gov.pl}}}

\author{B. Jakubczyk,\ \ G. Pietrzkowski\\
Institute of Mathematics, Polish Academy of Sciences\\
00-956 Warsaw 10, \'Sniadeckich 8, Poland}

\date{}

\begin{document}

\maketitle

\begin{abstract}
We study a separability problem suggested by mathematical
description of bipartite quantum systems. We consider hermitian
2-forms on the tensor product $H=K\otimes L$, where $K,L$ are
finite dimensional complex spaces.  Such a form is called
\emph{separable} if it is a convex combination of hermitian tensor
products $\conj{\sigma}_p\odot \sigma_p$ of 1-forms $\sigma_p$ on
$H$ that are product forms $\sigma_p=\varphi_p\otimes \psi_p$,
where $\varphi_p\in K^*$, $\psi_p\in L^*$.

We introduce an integral representation of separable forms. We
show that the integral of $\conj{D_{\conj{z}}\Phi}\od
D_{\conj{z}}\Phi$ of any square integrable map $\Phi:\C^n\to
\C^m$, with square integrable conjugate derivative
$D_{\conj{z}}\Phi$, is a separable form. Conversely, any separable
form in the interior of the set of such forms can be represented
in this way. This implies that any separable mixed state (and only
such states) can be either explicitly represented in the integral
form, or it may be arbitrarily well approximated by such states.
\end{abstract}


\section{Introduction}

Notions of separability and entanglement of states of a compound
quantum system are of vital importance in quantum physics and
quantum information theory. They emerged with the discovery of the
EPR effect \cite{einstein}, however, the throughout analysis came
much later \cite{werner, popescu, peres, horodecki, grabowski}.
Now the variety of theoretical problems where they play a central
role is constantly growing (quantum cryptography \cite{ekert},
quantum teleportation \cite{bennett}, dense coding
\cite{bennett2}, ...) and many theoretical properties are
confirmed in experiments \cite{mattle,bouwmeester,jennewein}.

Testing if a given state is separable or entangled (i.e.
non-separable) seems one of the central questions in the theory of
compound systems. Given a pure quantum state, it is easy to decide
if it is separable or entangled. For mixed states this is not the
case. At present there is no general and effective method to check
if a given mixed state is separable or entangled and the problem
seems hard \cite{horodecki2}. The most effective neccessary
condition is the partial transpose test \cite{peres, horodecki},
which is also sufficient in small dimensions \cite{horodecki}.

 In this work we present an indirect criterion for a mixed state
of a bi-partite system to be separable. We introduce integral
representations of separable states and prove that any state
having the integral representation is separable. Vice versa, any
state in the interior of the cone of separable states can be
represented in the integral form.

To be more precise, denote $\cH=\C^m\otimes(\C^n)^*$ and let
$d\mu$ be the standard Lebesgue measure in $\C^n\simeq\R^{2n}$.
With the use of identification $Hom(\C^n,\C^m)\simeq
\C^m\otimes(\C^n)^*$ our main results (Theorems \ref{hermitowskie}
and \ref{hermitowskie-t}) can be stated as follows.

\begin{theorem}
\label{separable} (a) For any square integrable map
$\Phi:\C^n\to\C^m$ having the conjugate differential
$D_{\conj{z}}\Phi(z)\in \cH$ square integrable, the density
operator (non-normalized mixed state) $\cH\to \cH$ defined by
\begin{eqnarray}
\label{int-sep} \int_{\C^n} | D_{\conj{z}}\Phi \rangle\, \langle
D_{\conj{z}}\Phi| \ d\mu(z)
\end{eqnarray}
is separable.

\noindent (b) Any separable mixed state in the interior of the set
of separable mixed states can be expressed in the above form.

\noindent (c) The above results also hold with $\C^n$ replaced
with the complex torus $\C\T^n$.
\end{theorem}

From mathematical view-point it is more convenient to state and
prove our results in terms of hermitian 2-forms. In particular,
using hermitian forms will not require the use of scalar product
in the statement of our results.

Indeed, positive semi-definite hermitian 2-forms can be used to
represent mixed states, instead of self-adjoint, positive
semi-definite operators on a Hilbert space. If $\cH$ is an
arbitrary Hilbert space, the obvious identification of these
notions is given by the formula
\begin{eqnarray*}
\langle w|\rho_o v \rangle = \rho_f(w,v),
\end{eqnarray*}
where $\rho_o$ is a self-adjoint operator in $\cH$ and
$\rho_f:\cH\times \cH \rightarrow \C$ is the corresponding
hermitian form, with respect to the scalar product $\langle\,
\cdot\, |\,\cdot\, \rangle$ in $\cH$. In our case of
$\cH=\C^m\otimes(\C^n)^*=Hom(\C^n,\C^m)$ we use the scalar product
$$
\langle A|B \rangle = \tr A^\dag B
$$
and then the identification, in the standard basis, is simply
given by
$$
(\rho_o)_{ijkl}=(\rho_f)_{ijkl}.
$$

\newpage

\section{Separable hermitian forms}

Let $H$ be a complex vector space. We will consider hermitian
2-forms on $H$, i.e., maps $\rho:H\times H\to \C$ which are
$\C$-linear with respect to the second argument and anti-linear
with respect to the first one. Given a linear function $f:H\to
\C$, we denote by $\conj{f}$ its complex conjugate,
$\conj{f}(z)=\conj{(f(z))}$, where in the latter case $\conj{}$
denotes complex conjugation in $\C$. Given linear functionals
$\alpha,\beta :H\to \C$, we define their hermitian tensor product
$\conj{\alpha}\oh \beta:H\times H\to\C$ by
$$
(\conj{\alpha}\oh \beta)(z,w) =
\frac{1}{2}(\conj{\alpha}(z)\beta(w)+\conj{\beta}(z)\alpha(w)),
$$
which is a hermitian 2-form. In our  considerations $H$ will be
the tensor product
$$
H = K \otimes L
$$
of complex vector spaces $K,L$ of finite dimensions.

\begin{definition}
A hermitian 2-form $\rho : H \times H \rightarrow \C$ is called
\emph{separable} if it can be expressed as
$$
\rho = \sum_{p=1}^{P} \conj{(\sigma^p)} \oh \sigma^p,
$$
where $P\geq 1$, $\sigma^p : H \rightarrow \C$ are linear
functionals (elements of $H^*$) such that
$$
\sigma^p = \varphi^p \otimes \psi^p,
$$
with $\varphi^p \in K^*$ and $\psi^p \in L^*$, and
$\conj{(\sigma^p)}$ denotes complex conjugation of $\sigma^p$.

A form $\rho$ is called \emph{product} form if $\rho =
\conj{\sigma}\oh \sigma$, where $\sigma = \varphi \otimes \psi$,
$\varphi \in K^*$ and $\psi \in L^*$. A positive semi-definite
$\rho$ is called \emph{entangled} if it is not separable.
\end{definition}

Note that separable hermitian 2-forms are positive  semi-definite.
The sets of separable (respectively, product) hermitian 2-forms on
$H$ will be denoted by $\cC_{sep}$ (resp. $\cC_{prod}$). These are
subsets of the real linear space $\cC$ of all hermitian 2-forms on
$H$. Note that if the sum defining $\rho$ is replaced by
$$
\rho=\sum_{p=1}^P \la_p\,  \conj{(\sigma^p)} \oh \sigma^p,
$$
with $\la_p\ge 0$ (equivalently, $\la_p\ge 0$ and $\sum_p\la_p=1$)
then we get an equivalent definition. Thus
$$
\cC_{sep}=\co \cC_{prod},
$$
where $\co A$ denotes the convex hull of $A$. Since $\dim
\cC=N^2$, where $N=\dim H$, it follows from the Carath\'eodory
theorem that in the above sums we can always take $P\le N^2$. The
following fact is well known (as it is crucial in further
considerations, we present its proof).

\begin{proposition}
\label{separ}
 The set of separable hermitian 2-forms is a closed, convex
cone with nonempty interior in the space of all hermitian 2-forms on $H$.
\end{proposition}
\begin{proof}
Convexity comes from the above remarks. To prove closedness we
consider the set $S$ of separable hermitian 2-forms
$$
\sum_{p=1}^{N^2} \la_p\,  \conj{(\varphi^p \otimes \psi^p)} \oh
(\varphi^p \otimes \psi^p),
$$
with $(\la_1,\ldots,\la_{N^2})$ in the closed simplex defined by
$\la_p\ge 0$ and $\sum_p\la_p=1$, and $\varphi^p, \psi^p$ in unit
spheres in $K^*, L^*$ (with respect to some fixed norms). The set
$S$ is a compact subset of the space of hermitian forms $\cC$, as
the image of a compact set under a suitable map. It does not
contain the zero form, as all such forms are nontrivial, positive
semi-definite. The cone $\cC_{sep}$ of all separable hermitian
2-forms is generated by $S$, thus $\cC_{sep}$ is closed.

To prove that $\cC_{sep}$ has nonempty interior in $\cC$, it is
enough to show that there is no nontrivial linear functional
acting on $\cC$ which annihilates $\cC_{sep}$. To begin with, let
us fix hermitian products in $K^*$ and $L^*$, and orthonormal
basis $\e_1,\ldots,\e_n$ in $K^*$ and $\f_1,\ldots,\f_m$ in $L^*$
($n = \dim K, m=\dim L$) with respect to these products. Define
two sets of vectors in $K^*$ and $L^*$
\begin{eqnarray*}
K_0 = \{\e_a + e_K\e_b \ |\ a,b=1,\ldots,n, \ e_K = 1,i \}\subset K^*, \\
L_0 = \{\f_c + e_L\f_d \ |\ c,d=1,\ldots,m, \ e_L = 1,i \}\subset
L^*. \,
\end{eqnarray*}
Since $\conj{(\varphi \otimes \psi)} \oh (\varphi \otimes \psi)
\in \cC_{sep}$, for $\varphi\in K_0$, $\psi\in L_0$, it is enough
to show that
\begin{quotation}\noindent
if an element $\theta\in\cC^*$ of the dual space $\cC^*$ vanishes
on every element $\conj{(\varphi \otimes \psi)} \oh (\varphi
\otimes \psi)$, with $\varphi \in K_0$ and $\psi\in L_0$, then
$\theta\equiv 0$.
\end{quotation}
Denote $\theta_{ijkl} = \theta(\conj{(\e_i \otimes \f_j)} \oh
(\e_k \otimes \f_l))$, so that $\theta_{ijkl} =
\conj{\theta}_{klij}$. We will successively show that
$\theta_{ijkl}=0$ for all $i,k=1,\ldots,n$ and $j,l=1,\ldots,m$.
In every step we will be using the identities from the previous
steps. Note first that if $a=b$, $c=d$ and $e_K=e_L=1$ (i.e. $\e_a
+ \e_a \in K_0, \f_c + \f_c \in L_0$) then
\begin{eqnarray*}
0 = \theta(\conj{(2\e_a \otimes 2\f_c)} \oh (2\e_a \otimes 2\f_c))
 = 16\theta_{acac},
\end{eqnarray*}
and thus $\theta_{ijij}=0$. Next, if we take  $a\neq b$, $c=d$ and
$e_L=1$ then, using hermicity of $\theta$, we obtain
\begin{eqnarray*}
0 & = & \theta\Big(\conj{\big((\e_a +  e_K\e_b)\otimes 2\f_c\big)} \oh
      \big((\e_a + e_K\e_b)\otimes 2\f_c\big)\Big) = 8\Re(e_K\theta_{acbc}),
\end{eqnarray*}
therefore $\theta_{ijkj}=0$, since real and imaginary parts of it
vanish. Analogously we prove that $\theta_{ijil}=0$. Finally, if
we take $a\neq b$, $c\neq d$,  we obtain
\begin{eqnarray*}
0 & = & \theta\Big(\conj{\big((\e_a +  e_K\e_b)\otimes (\f_c + e_L\f_d)\big)}
 \oh \big((\e_a + e_K\e_b)\otimes (\f_c + e_L\f_d)\big)\Big) \\
  & = & 2(\Re(e_Ke_L\theta_{acbd}) + \Re(e_Ke^*_L\theta_{adbc})).
\end{eqnarray*}
Since four combinations of $e_K$ and $e_L$ give independent linear
equations for real and imaginary parts of $\theta_{acbd}$ and
$\theta_{adbc}$, we conclude that  $\theta_{ijkl}=0$.
\end{proof}

\section{Integral representations}

Let $H_1,H_2$ be vector spaces over $\C$  and let
$$
H=H_1\ot (H_2)^*=\Hom(H_2,H_1).
$$
The dual space $H^*=(H_1)^*\ot H_2\simeq H_2\ot (H_1)^*$ can be
identified with the space of maps $\Hom(H_1, H_2)$ and then the
duality product is given by
$$
\langle A,B\rangle=\tr (AB)=\tr (BA), \quad A\in H^*,\ B\in H.
$$

Given a $\C$-linear map $A:H_1\to H_2$, we define a hermitian form
$A\od A$ on $H=\Hom(H_2,H_1)$, which is the hermitian product of
two copies of the linear functional $B\to \tr(AB)$,
$$
(A\od A)(B,C)= \conj{\tr (BA)}\, \tr (AC),
$$
where $B,C\in \Hom(H_2,H_1)$. This form is also given by the
bilinear extension of $(A\od A)(v\otimes w,\t v\otimes \t w)=
\conj{(wAv)}\, \t wA\t v$, for $w,\t w\in H_2^*$ and $v,\t v\in
H_1$.

For a complex variable $z=x+iy$ and its complex adjoint
$\conj{z}=x-iy$ we use the usual notation $dz=dx+idy$,
$d\conj{z}=dx-idy$ for te complex 1-forms and $\partial_z =
(\partial_x - i\partial_y)/2$ and $\conj{\partial}_z =(\partial_x
+ i\partial_y)/2$ for the dual complex vector fields. Then
$d\conj{z}\wedge dz=2idx\wedge dy$.

Let us assume that $H_1 = \C^n$ and $H_2 = \C^m$. We shall
consider a map $\Phi = (\Phi_1,\ldots,\Phi_m) : \C^n \rightarrow
\C^m$ which is square integrable and, as a map from $\R^{2n}$ to
$\R^{2m}$, it has weak differential which is square integrable,
too. We denote by $D_{\conj{z}}\Phi(z) \in \C^m\otimes(\C^n)^*$
the conjugate differential of $\Phi$ at $z$ which, by definition,
is the complex linear  map defined by the complex matrix
\begin{eqnarray*}
(D_{\conj{z}}\Phi(z))_{ij} = \conj{\partial}_{z_j} \Phi_i
(z_1,\ldots,z_n),
\end{eqnarray*}
where $\conj{\partial}_{z_j}=(\partial_{x_j}+i\partial_{y_j})/2$.
The linear map $D_{\conj{z}}\Phi(z) : \C^n \rightarrow \C^m$ can
be considered as an element of the dual space $H^*$ to the tensor
product
$$
H = \C^n\otimes(\C^m)^*.
$$
 Denote
$$
d\conj{z}\wedge dz = d\conj{z}_1\wedge dz_1 \wedge \ldots \wedge
d\conj{z}_n \wedge dz_n.
$$

\begin{theorem}
\label{hermitowskie} (a) For an arbitrary square integrable map
$\Phi:\C^n\to\C^m$ with square integrable conjugate differential
$D_{\conj{z}}\Phi(z)$, the hermitian 2-form $\rho_\Phi:H\times
H\to\C$ defined by
\begin{eqnarray}
\label{przedstawienie} \rho_\Phi=\frac{1}{(2i)^n}\int_{\C^n}
\conj{(D_{\conj{z}}\Phi(z))} \oh D_{\conj{z}}\Phi(z) \
d\conj{z}\wedge dz
\end{eqnarray}
is separable.

\noindent (b) Any separable hermitian 2-form in the interior of
the set of separable hermitian 2-forms can be expressed in the
above form.
\end{theorem}

The same result holds with $\C^n$ replaced by the complex torus.
Recall that the $n$-dimensional complex torus is the quotient
group
$$ \C\T^n = \C^n\slash\Lambda, $$
with topology and Lebesgue measure inherited from $\C^n$, where
$\Lambda$ is the lattice
$$
\Lambda = \{(2\pi (a_1 + i b_1),\ldots,2\pi (a_n + i b_n)) \in\C^n
\ |\  a_k,b_k\in\Z, k=1,\ldots,n \}
$$
in $\C^n$. Given two points $z,\tilde{z}\in \C\T^n$, there is a
natural identification of the tangent spaces $T_z\C\T^n$ and
$T_{\tilde{z}}\C\T^n$ via the standard parallel shift in $\C^n$.
Therefore, as earlier, for any mapping $\Phi:\C\T^n\to \C^m$
having the weak differential $D_{\conj{z}}\Phi$ the linear map
$D_{\conj{z}}\Phi(z) : \C^n \rightarrow \C^m$ can be considered as
element of the dual space $H^*= (\C^n)^*\otimes\C^m$.

\begin{theorem}
\label{hermitowskie-t} Statements (a) and (b) of Theorem
\ref{hermitowskie} hold if we replace $\C^n$ with $\C\T^n$, i.e.,
for square integrable maps $\Phi:\C\T^n\to\C^m$, with square
integrable conjugate differential $D_{\conj{z}}\Phi$, and the
integral is taken over $\C\T^n$.
\end{theorem}

Both theorems will be proved in the following two sections.  We
will also show (Theorem \ref{hermitowskie2}) that not all
hermitian, positive semi-definite forms are integrally
representable.  Such forms lie in the boundary of the cone of all
separable forms.


\section{Separability of $\rho_\Phi$}

 In proving statements (a) of both theorems we will use
Fourier transform and the Hahn-Banach theorem.

The following elementary facts will be used in the proof. Given a
function $f: \C \rightarrow \C$, we can write it as a complex
valued function $\R^2\to\C$ by identifying $f(x+iy)=f(x,y)$.
Assuming that it is differentiable at $z=x+iy$, we have
\begin{equation}
\label{equivalence}
 \conj{\partial}_z f(x+iy)  =
\frac{1}{2}(\partial_x f(x,y) + i\partial_y f(x,y)).
\end{equation}
Consider the Fourier transform of $f:\R^2\to \C$,
\begin{eqnarray*}
\widehat{f}(\xi,\zeta)
  =  \frac{1}{{2\pi}}\int_{\R^{2}}e^{-i( x\xi + y\zeta)}f(x,y) \ dxdy.
\end{eqnarray*}
Then integration by parts gives
\begin{eqnarray*}
\widehat{\partial_{x}f}(\xi,\zeta) = i\xi\widehat{f}(\xi,\zeta),
\qquad \widehat{\partial_{y}f}(\xi,\zeta) =
i\zeta\widehat{f}(\xi,\zeta).
\end{eqnarray*}
Using (\ref{equivalence}) and taking $\kappa=\xi+i\zeta$, we
aggregate this in the complex expression
\begin{eqnarray}
\label{a} \widehat{\conj{\partial}_zf}(\kappa) & = &
\frac{1}{2}\left(\widehat{\partial_{x}f}(\xi,\zeta)
+ i \widehat{\partial_{y}f}(\xi,\zeta)\right) \nonumber \\
& = & \frac{1}{2}i\left(\xi\widehat{f}(\xi,\zeta) +
i\zeta\widehat{f}(\xi,\zeta)\right)
 = \frac{1}{2}i\kappa\widehat{f}(\kappa).
\end{eqnarray}

Let $X$ be a finite dimensional vector space (or more generally, a
Banach space). We shall need the following property, which follows
from the Hahn-Banach theorem by a standard separation argument.

 \begin{proposition}\label{char}
Let  $S\subset X$ be a subset and $C\subset X$ be the smallest
convex cone containing $S$. If $C$ is closed and $x_0\in X$ is an
element satisfying
$$
\langle y,x_0\rangle \ge0, \quad \text{for any $y\in X^*$ such
that $\langle y,x\rangle \ge0$ for all $x\in S$},
$$
then $x_0$ lies in $C$.
 \end{proposition}

We will concentrate on a linear space $\cC$ over $\R$ of all
hermitian 2-forms $\rho : H\times H \rightarrow \C$ and its dual
$\cC^*$, with the duality product
$\langle\cdot,\cdot\rangle:\cC^*\times\cC \rightarrow \R$. Take
any linear coordinates in $K$ and $L$. We have associated
coordinates in $\cC$ and dual coordinates in $\cC^*$. We can
express $\theta = (\theta_{ijkl}) \in \cC^*$ and $\rho =
(\rho_{ijkl}) \in \cC$, in these coordinates, where
$\conj{\theta_{ijkl}} = \theta_{klij}$ and $\conj{\rho_{ijkl}} =
\rho_{klij}$. Then
$$
\langle\theta,\rho\rangle = \sum_{ijkl} \theta_{ijkl}\rho_{ijkl}.
$$

{\bf Proof of Theorem \ref{hermitowskie} (a).}
 Recall that $\cC_{prod}$ denotes the set of product hermitian 2-forms,
and $\cC_{sep}$ the set of separable hermitian 2-forms. Since the
convex hull of $\cC_{prod}$ is equal to $\cC_{sep}$ and it is a
closed cone in $\cC$ then, by Proposition \ref{char}, $\rho_\Phi
\in \cC_{sep}$ if $\langle\theta,\rho_\Phi\rangle \geq 0$ for all
$\theta \in \cC^*$ such that
$$
\langle\theta,\rho\rangle \geq 0, \qquad \text{for any
}\rho\in\cC_{prod}.
$$

Denote, for brevity, $\conj{\partial}_j = \conj{\partial}_{z_j}$.
For $\theta \in \cC^*$ in the conjugate cone to the cone of
separable states and  $\Phi = (\Phi_1,\ldots,\Phi_m)$ we can write
\begin{eqnarray*}
\langle\theta,\rho_\Phi\rangle
     =  \sum_{ijkl}\theta_{ijkl} \frac{1}{(2i)^n}\int_{\C^n}
     \conj{(\cpar_j\Phi_i(z))}\,\cpar_l\Phi_k(z) \ d\conj{z}\wedge dz .\\
\end{eqnarray*}
Using Fourier transform and Perseval's equality $\int
f^*gd\conj{z}\wedge dz=\int \widehat{f}^*\widehat
gd\conj{\kappa}\wedge d\kappa$ we get
$$
\int_{\C^n} \conj{(\cpar_j\Phi_i(z))}\,\cpar_l\Phi_k(z) \
d\conj{z}\wedge dz =
 \int_{\C^n} \conj{(\widehat{\conj{\partial}_j\Phi_i}(\kappa))}\
 (\widehat{\conj{\partial}_l\Phi_k}(\kappa)) \ d\conj{\kappa}\wedge
 d\kappa.
$$
Thus (\ref{a}), and the fact that $\theta$ is positive on product
states gives
\begin{eqnarray*}
\langle\theta,\rho_\Phi\rangle & = &
 \frac{1}{4(2i)^n}\int_{\C^n}\sum_{ijkl}\theta_{ijkl}
 \conj{(i\kappa_j\widehat{\Phi}_i(\kappa))}
 (i\kappa_l\widehat{\Phi}_k(\kappa)) \ d\conj{\kappa}\wedge
 d\kappa \geq 0,
\end{eqnarray*}
which ends the proof.
\endproof

\medskip

 {\bf Proof of Theorem \ref{hermitowskie-t} (a).}
 The proof is analogous to the previous one. The only thing that one has
 to observe is the following. Given a square integrable function
 $f:\C\T^n\to\C$, its Fourier transform $\hat{f}:\Z^{n}\times\Z^{n}\to\C$
 is given by the formula
\begin{eqnarray*}
\hat{f}(\alpha ,\beta)
  =  \frac{1}{(4\pi i)^n}\int_{\C \T^n}
  e^{-i( \langle x,\alpha\rangle + \langle y,\beta\rangle)}f(z)
  \ d\conj{z}\wedge dz,
\end{eqnarray*}
where $z=x+iy$ mod $2\pi(\Z+ i\Z)$ are points in $\C\T^n$,
$\langle x,\alpha\rangle = \sum x_i \alpha_i$ and $\langle
y,\beta\rangle = \sum y_i \beta_i$. Then integration by parts
gives
\begin{eqnarray*}
\widehat{\conj{\partial}_kf}(\alpha,\beta)
 = \frac{1}{2}i(\alpha_k+i\beta_k)\widehat{\Phi}(\alpha,\beta),
\end{eqnarray*}
and the Perseval's equality reads as
$$
\frac{1}{(2i)^n} \int_{\C\T^n}  \conj{(f(z))}g(z) \
d\conj{z}\wedge dz = \sum_{(\alpha,\beta)\in\Z^{2n}}
\conj{(\hat{f}(\alpha,\beta))} (\hat{g}(\alpha,\beta)),
$$
where $g:\C\T^n\to\C$ is another square integrable function.
\endproof

\section{Construction of $\Phi$ in Theorem \ref{hermitowskie}(b)}

The construction of maps $\Phi$ which produce or approximate an
arbitrary separable hermitian form is divided into three steps.
First we will prove that any product hermitian 2-form can be
arbitrarily closely approximated by the hermitian 2-forms
$\rho_\Phi$. Then we approximate any separable hermitian 2-form.
Finally, we prove that any separable hermitian 2-form in the
interior of all separable hermitian 2-forms can be expressed in
the integral form (\ref{przedstawienie}).

\vspace{0,5cm} \noindent\textbf{Step 1. Approximation of product
forms.}

\noindent Consider the function $f_w:\C^n\to\C$ given by
\begin{eqnarray}
\label{f} f_w(z) =  h_w(z)\,g_\alpha(z),
\end{eqnarray}
where  $w\in\C^n$ and $\alpha > 0$ are fixed and
$$
h_w(z)=e^{\langle z,w\rangle - \langle w,z\rangle} = e^{2i \Im
\langle z,w\rangle}, \qquad g_\alpha(z)= (\pi \alpha)^{-n/2}\
e^{-\langle z,z\rangle/ 2\alpha},
$$
with $\langle z,w\rangle=\sum \conj{z}_jw_j$. Clearly,
$$
\conj{\partial}_jf_w(z)=(w_j-\frac{z_j}{2\alpha})f_w(z).
$$
 We have $|h_w(z)|=1$ and, writing $z_j=x_j+iy_j$,
$$
|f_w(z)|= \frac{1}{(\pi \alpha)^{n/2}}
e^{-\sum(x_j^2+y_j^2)/2\alpha}.
$$
This allows to prove the following lemma.

\begin{lemma}
\label{gauss}
$$
 \frac{1}{(2i)^n}\int_{\C^n} \conj{(\conj{\partial}_j f_w(z))}\conj{\partial}_l f_w(z) \ d\conj{z}\wedge dz =
 \conj{w}_j w_l +\frac{1}{4\alpha}\delta_{j_l}.
$$
\end{lemma}

\proof
 We have (all integrals are taken with respect to the measure
 $d\conj{z}\wedge dz$)
\begin{eqnarray*}
\int_{\C^n} \conj{(\conj{\partial}_j f_w(z))}\conj{\partial}_l
f_w(z)
 & = & \int_{\C^n} \conj{(w_j-z_j/2\alpha)}(w_l -z_l/2\alpha) |f_w(z)|^2 \\
 & = & \conj{w}_jw_l \int_{\C^n} |f_w(z)|^2
       + \int_{\C^n} \frac{\conj{z}_jz_l}{4\alpha^2}|f_w(z)|^2 \\
 & - & \frac{\conj{w}_j}{2\alpha}\int_{\C^n} z_l |f_w(z)|^2
 - \frac{w_l}{2\alpha}\int_{\C^n} \conj{z}_j |f_w(z)|^2 \\
 & = & (2i)^n(\conj{w}_j w_l +\frac{1}{4\alpha}\delta_{jl}),
\end{eqnarray*}
where the first integral is computed using the Fubini theorem in
the form
\begin{eqnarray*}
\int_{\C^n}|f_w|^2
 & = & \frac{1}{(\pi \alpha)^{n}}\prod_{s=1}^n \int_{\C^1}e^{-(x_s^2+y_s^2)/\alpha}d\conj{z}_s\wedge
 dz_s\\
 & = & \Big(\frac{2i}{\pi \alpha}\Big)^n \prod_{s=1}^n \int_{\R^1}e^{-x_s^2/\alpha}dx_s
\int_{\R^1}e^{-y_s^2/\alpha}dy_s
\end{eqnarray*}
and the standard integral
\begin{eqnarray}
\label{gauss0} \int_{\R} e^{-x^2/\alpha}  \ dx = \sqrt{\pi\alpha},
\end{eqnarray}
while in computing the second one with $j=l$ we use the above
integral and
\begin{eqnarray}
\label{gauss2} \int_{\R} x^2e^{-x^2/\alpha}  \ dx =
\sqrt{\pi\alpha}\ \frac{\alpha}{2}.
\end{eqnarray}
The third and the fourth integrals are equal to zero as, after
applying the Fubini theorem we obtain a factor of the form
$\int_\C z_l e^{-|z_l|^2/\alpha} d\conj{z}_l\wedge
dz_l$, which is zero since the real and imaginary parts of the
integrated function are antisymmetric with respect to
corresponding (real or imaginary) axes. The same argument implies
vanishing of the second integral, when $j\ne l$.
 \endproof

We shall approximate the hermitian 2-form
$\conj{\sigma}\oh\sigma$, where $\sigma = \varphi\otimes\psi$ and
$\varphi \in \C^m$, $\psi \in (\C^n)^*$. To do that, let us take
$\alpha>0$ and consider the maps $\Phi:\C^n\to\C^m$ given by
\begin{eqnarray*}
\Phi(z) =\varphi \ \frac{1}{(\pi\alpha)^{n/2}} e^{\langle
z,\psi\rangle - \langle \psi,z\rangle}e^{-\langle z,z\rangle/
2\alpha} =\varphi f_\psi(z),
\end{eqnarray*}
where $f_\psi$ is the function defined in (\ref{f}), with
$w=\psi$, and we denote $\langle \psi,z \rangle=\sum
\conj{\psi}_iz_i$. By Lemma \ref{gauss} we have
\begin{eqnarray*}
(\rho_\Phi)_{ijkl} & = & \conj{\varphi}_i \varphi_k
\frac{1}{(2i)^n}\int_{\C^n} \conj{(\cpar_j{f_\psi(z)})}
\cpar_l{f_\psi(z)} \ d\conj{z}\wedge dz
\\
& = &  \conj{\varphi}_i \varphi_k (\conj{\psi}_j \psi_l +
\frac{1}{4\alpha} \delta_{jl}).
\end{eqnarray*}
We see that, with $\alpha$ sufficiently large,  the hermitian
2-form $\conj{\sigma}\oh\sigma$ can be arbitrarily closely
approximated by $\rho_\Psi$.
\endproof

\medskip \noindent\textbf{Step 2. Approximation of separable
forms.}

\noindent To prove that any separable state can be arbitrarily
closely approximated by states in the integral form we use the
result for product states. We start with a technical lemma.

\begin{lemma}
\label{inter2} For $f_v$ and $f_w$ of the form (\ref{f}) we have
\begin{eqnarray*}
\lefteqn{ \frac{1}{(2i)^n}\int_{\C^n} \conj{(\conj{\partial}_j
f_v(z))}\conj{\partial}_l f_w(z) \ d\conj{z}\wedge dz =
} \\
&& = \frac{1}{4}\skladnik{v}{w}.
\end{eqnarray*}
\end{lemma}
\begin{proof}
Note that
\begin{eqnarray*}
\int_{\C^n} \conj{(\conj{\partial}_j f_v(z))}\conj{\partial}_l
f_w(z) & = & \conj{v}_j w_l \int_{\C^n} \conj{f}_v(z) f_w(z) +
\int_{\C^n} \frac{\conj{z}_j z_l}{4\alpha} \conj{f}_v(z) f_w(z)
\\
&& - \frac{\conj{v}_j}{2\alpha} \int_{\C^n} z_l \conj{f}_v(z)
f_w(z) - \frac{w_l}{2\alpha} \int_{\C^n} \conj{z}_j \conj{f}_v(z)
f_w(z),
\end{eqnarray*}
where all integrals are taken with respect to $d\conj{z}\wedge
dz$. With $z_j=x_j+iy_j$ we have
\begin{eqnarray*}
\conj{f}_v(z) f_w(z)
 & = & \conj{h}_v(z)h_w(z)g^2_\alpha(z) =
       e^{-2i \Im \langle z,v\rangle}  e^{2i \Im \langle z,w\rangle}
       g^2_\alpha(z)\\
 & = & (\pi\alpha)^{-n} \prod_{s=1}^n
       e^{2ix_s\Im(w_s-v_s)} e^{-2iy_s\Re(w_s-v_s)}e^{-(x_s^2 +
       y_s^2)/\alpha}.
\end{eqnarray*}
 Thus the proof of the lemma is a straightforward
calculation analogous to the calculations in the proof of Lemma
\ref{gauss}, with the use of three additional integrals
\begin{eqnarray*}
\int_{\R} e^{-x^2/\alpha} e^{i\gamma x}  \ dx & = & \sqrt{\pi\alpha} \ e^{-\frac{\alpha\gamma^2}{4}}, \\
\int_{\R} xe^{-x^2/\alpha} e^{i\gamma x}  \ dx & = & \sqrt{\pi\alpha} \ \frac{i\alpha\gamma}{2}  e^{-\frac{\alpha\gamma^2}{4}}, \\
\int_{\R} x^2e^{-x^2/\alpha} e^{i\gamma x}  \ dx & = &
\sqrt{\pi\alpha} \ \frac{(2\alpha - \alpha^2\gamma^2)}{4}
e^{-\frac{\alpha\gamma^2}{4}}.
\end{eqnarray*}
\end{proof}

Now consider a separable hermitian 2-form
$$
\rho = \sum_{p=1}^{P} \conj{(\sigma^p)} \oh \sigma^p, \\
$$
where $\sigma^p = \varphi^p \otimes \psi^p$, $\varphi^p\in\C^m$,
$\psi^p\in(\C^*)^n$ and we additionally require that $\psi^p\neq
\psi^q$ for $p\neq q$, and $P\geq 1$. The set of such 2-forms is
dense in the set of all separable 2-forms. Thus, it is sufficient
to approximate such 2-forms. To do that we consider the map
$\Phi:\C^n\to\C^m$,
\begin{eqnarray}
\label{Phi} \Phi(z) = \sum_{p=1}^{P} \varphi^p f_{\psi^p} =
\frac{1}{(\pi\alpha)^{n/2}} \sum_{p=1}^{P}\varphi^p \ e^{\langle
z,\psi^p\rangle-\langle\psi^p,z\rangle} e^{-\langle z,z\rangle/
2\alpha},
\end{eqnarray}
where we use  formula (\ref{f}) for $f_{\psi^p}$'s. By Lemma
\ref{inter2} we have
\begin{eqnarray*}
(\rho_\Phi)_{ijkl} & = & \sum_{p,q} \conj{(\varphi_i^p)}
\varphi_k^q \frac{1}{(2i)^n}\int_{\C^n}
\conj{(\cpar_j{f_{\psi_p}(z)})} \cpar_l{f_{\psi_q}(z)} \
d\conj{z}\wedge dz
\\
& = & \frac{1}{4}\sum_{p,q} \conj{(\varphi^p_i)} \varphi_k^q
\skladnik{\psi^p}{\psi^q} .
\end{eqnarray*}
Since $e^{-\alpha|\psi^p-\psi^q|^2}\to 0$, unless $p=q$,  this
expression converges to
\begin{eqnarray*}
\sum_{p=1}^P \conj{(\varphi^p_i)} \varphi_k^p
\conj{(\psi_j^p)}\psi_l^p
\end{eqnarray*}
when $\alpha$ tends to infinity. Thus we can approximate the form
$\rho$ with integrally representable forms $\rho_\Psi$.
\endproof

From the above proof we can easily obtain the first part of the
following proposition.

\begin{proposition}
\label{interior} For any $ \ P\in\N,\ \varphi_1,\ldots,\varphi_P
\in \C^m$, $\psi_1,\ldots,\psi_P \in (\C^n)^*$ and $\alpha >0$,
the hermitian 2-form $\rho$ with coefficients
\begin{eqnarray}
\label{formy} \rho_{ijkl}  = \frac{1}{4}\sum_{p, q}
\conj{(\varphi^p_i)} \varphi_k^q \skladnik{\psi^p}{\psi^q}
\end{eqnarray}
is an integrally representable separable hermitian 2-form and
$\rho=\rho_\Phi$ with $\Phi$ given in (\ref{Phi}). Moreover, every
hermitian 2-form in the interior of the cone of all separable
2-forms is of the form (\ref{formy}).
\end{proposition}

The second part of this proposition is equivalent to the statement
of Theorem \ref{hermitowskie} (b). Therefore, to end the proof of
the theorem we need to prove the proposition. Before we do that we
make a few remarks about integrally representable hermitian
2-forms, which follow from (\ref{formy}).

\begin{remark}\label{rem-5.4}
Assume that $\varphi_p = \varphi$ for  $p = 1,\ldots,P$. Then
(\ref{formy}) reduces to
\begin{eqnarray*}
\rho_{ijkl} = \frac{1}{4}\conj{\varphi}_i \varphi_k \sum_{p, q}
\skladnik{\psi^p}{\psi^q}
\end{eqnarray*}
and $\rho$ is of the form
$(\conj{\varphi}\oh\varphi)\ot\tilde{\rho}$, where $\tilde{\rho}$
is a hermitian 2-form of rank $n$ on $\C^n$ (the rank of $\rho$ is
at least $n$, as we will see in Proposition \ref{rank}).
\end{remark}

\begin{remark}\label{rem-5.5}
If we take $\psi_p = \psi$, for $p = 1,\ldots,P$, then the sum
defining $\Phi$ in (\ref{Phi}) has the same exponential function
and it reduces to one summand (this is the simplest case
considered in the first step of the proof). By the same reason we
see that there is no loss of generality to assume in (\ref{Phi})
that $\psi_p \neq \psi_q$, if $p\neq q$.
\end{remark}

\medskip
\noindent\textbf{ Step 3. Proof of the second part of Proposition
\ref{interior}.}

\noindent Consider a separable hermitian 2-form $\rho$ in the
interior of $\cC_{sep}\subset \cC$,
 where $\cC$ is  the real
vector space of hermitian 2-forms. It follows from Proposition
\ref{separ} that we can find $D=\dim \ \span\cC_{sep}=(nm)^2$ linearly
independent separable hermitian 2-forms $\rho^1,\ldots,\rho^D$,
considered as vectors in the space $\cC$ of hermitian forms, such
that $\rho = \sum_{d=1}^D \lambda^\rho_d \rho^d$ with
$\lambda^\rho_1 > 0, \ldots, \lambda^\rho_D > 0$.
Without loss of generality (changing slightly $\rho^d$, if
necessary) we can assume that
$$
\rho^d = \sum_{p=1}^{P_d} \conj{\big(\sigma^\pd\big)} \oh \sigma^\pd,
$$
where $P_d\geq 1$, $\sigma^\pd = \varphi^\pd \otimes \psi^\pd$,
with $\varphi^\pd\in\C^m$, $\psi^\pd\in(\C^n)^*$ and $\psi^\pc\neq
\psi^\qd$ for $(p,c)\neq (q,d)$. For each $d$ consider a mapping
$\Phi_d^\alpha:\C^n\to\C^m$ expressed by (\ref{Phi}) with
appropriate $\varphi$'s and $\psi$'s. Define the map
$\Phi^\alpha:\R_+^{D}\times\C^n\to\C^m$,
\begin{eqnarray*}
\Phi^\alpha(\lambda_1,\ldots,\lambda_D,z) = \sum_{d=1}^D
\sqrt{\lambda_d} \Phi_d^\alpha(z).
\end{eqnarray*}
 Take $\alpha=1/\beta^2$. We introduce the mapping
$\Upsilon:\R_+^D\times\R \to \cC\simeq\R^D$ defined by
\begin{eqnarray*}
\lefteqn{\Upsilon(\lambda_1,\ldots,\lambda_D,\beta) = } \\
& = & \left\{\begin{array}{ll}
\frac{1}{(2i)^n}\int_{\C^n}
 \conj{(D_{\conj{z}}\Phi^{\beta^{-2}}(\lambda,z))}
 \oh D_{\conj{z}}\Phi^{\beta^{-2}}(\lambda,z)
 \ d\conj{z}\wedge dz, & \beta \neq 0 \\
\rho, & \beta = 0
\end{array}\right. .
\end{eqnarray*}
From formula (\ref{formy}) obtained in Step 2 we know that
\begin{eqnarray*}
\Upsilon(\lambda,\beta) & = &
\sum_{d=1}^{D}\sum_{p=1}^{P_d} \lambda_d \conj{\big(\sigma^\pd\big)}
 \oh \sigma^\pd + R(\lambda,\beta) \\
& = & \sum_{d=1}^{D} \lambda_d \rho_d + R(\lambda,\beta),
\end{eqnarray*}
where $R:\R_+^D\times\R\to\cC$ is a differentiable mapping of
class $C^\infty$, given in coordinates by
\begin{eqnarray*}
\lefteqn{ R_{ijkl}(\lambda,\beta)  = \frac{1}{4}
 \sum_{d=1}^D \sum_{p=1}^{P_d} \lambda_d\conj{\big(\varphi_i^\pd\big)} \varphi_k^\pd \beta^2\delta_{jl} + } \\
&& + \frac{1}{4}\sum_{(p,c)\neq (q,d)} \sqrt{\lambda_c\lambda_d} \conj{\big(\varphi^\pc_i\big)} \varphi_k^\qd
\\
&& \qquad \times\Big(\conj{\big(3\psi_j^\pc - \psi_j^\qd\big)}\big(3\psi_l^\qd - \psi_l^\pc\big) +
{\beta^2\delta_{jl}}\Big) e^{-\frac{1}{\beta^2}|\psi^\pc - \psi^\qd|^2}.
\end{eqnarray*}

Since $\Upsilon(\lambda^\rho,0)=\rho$, we can complete the proof
by using the implicit function theorem. We only have to prove that
the rank of the differential of $\Upsilon$ with respect to
$\lambda$ is maximal  at $(\lambda^\rho,0)$. But this is trivial
since $R(\cdot,0)=0$ ($R(\cdot,\beta)$ converges locally uniformly
to zero function when $\beta$ tends to zero) and
$\rho^1,\ldots,\rho^D$ is a basis in $\cC\simeq R^D$.
 \endproof

\section{Construction of $\Phi$ in Theorem \ref{hermitowskie-t} (b)}

 Consider the family of mappings $\Phi:\C\T^n\to\C^m$ of the form
\begin{eqnarray}
\label{Tprzyblizanie} \Phi(z) = 2\sum_{p=1}^{P}
c_p^{-1}\varphi^p\, \chi_{a^p,b^p}(z),
\end{eqnarray}
where $z = (z_1,\ldots,z_n)$ with $z_j = x_j + iy_j \mod
2\pi(\Z+i\Z)$, $P\in\N$, and for all for $p=1,\ldots,P$ we take
$\varphi^p\in\C^m$, $(a^p,b^p)\in\Z^n\times\Z^n$, $c_p\in\N$ and
\begin{eqnarray*}
\chi_{a^p,b^p}(z) = (2\pi)^{-n}e^{i(\langle x, a^p\rangle +
\langle y, b^p\rangle)}.
\end{eqnarray*}
Without loosing generality we will assume that $(a^p,b^p)\ne
(a^q,b^q)$ if $p\neq q$. Then, with  $\psi^p = \frac{1}{c^p}(a^p +
ib^p)\in\Q^n+i\Q^n$ and $\varphi^p\in\C^m$, we have

\begin{lemma}
\label{Tlemat}
$$
\rho_\Phi = \sum_{p=1}^{P} \conj{(\varphi^p\ot \psi^p)}\oh
 (\varphi^p\ot \psi^p).
$$

\end{lemma}
\begin{proof}
Clearly,
\begin{eqnarray*}
\cpar_j\Phi(z) = \sum_{p=1}^{P} ic_p^{-1}\varphi^p(a_j^p+ib_j^p)
\chi_{a^p,b^p}(z),
\end{eqnarray*}
where $\cpar_j = \frac{1}{2}(\partial_{x_j} + i \partial_{y_j})$.
Since
\begin{eqnarray*}
\frac{1}{(2i)^n}\int_{\C\T^n}
\conj{(\chi_{a^p,b^p}(z))}{\chi_{a^q,b^q}(z)} \ dz\wedge d\conj{z}
= \left\{\begin{array}{ll}
1, & \mathrm{if}\ p=q, \\
0, & \text{otherwise},
\end{array}\right.
\end{eqnarray*}
we find that
\begin{eqnarray*}
(\rho_\Phi)_{ijkl} & = & \frac{1}{(2i)^n}\int_{\C\T^n} \conj{(\cpar_j\Phi_i(z))}\cpar_l\Phi_k(z) \ d\conj{z}\wedge dz = \\
& = & \sum_{p=1}^{P}
c_p^{-2}\conj{(\varphi_i^p)}\conj{(a_j^p+ib_j^p)}
\varphi^p_k(a_l^p+ib_l^p).
\end{eqnarray*}
Without coordinates one can write it as
\begin{eqnarray}
\label{Togolny} \rho_\Phi = \sum_{p=1}^{P} \conj{(\varphi^p\ot
\psi^p)}\oh (\varphi^p\ot \psi^p),
\end{eqnarray}
where $\psi^p = \frac{1}{c^p}(a^p + ib^p)$.
\end{proof}

\begin{remark}
Note that the set of forms
$$
\sum_{p=1}^{P} \conj{(\varphi^p\ot \psi^p)}\oh (\varphi^p\ot
\psi^p),
$$
with arbitrary $P\geq 1$, $\varphi^p\in\C^m$, and $\psi^p\in(\Q +
i\Q)^n$, is convex.
\end{remark}

 \noindent {\it Proof of Theorem \ref{hermitowskie-t} (b).}
It follows from the lemma that every separable hermitian 2-form
can be arbitrarily closely approximated by the forms $\rho_\Phi$,
since the set of $\psi^p$ of the above form is dense in $\C^n$.
This, together with the fact that the set of such forms is convex,
implies that every hermitian 2-form in the interior of $\cC_{sep}$
is integrally representable.
\endproof

\begin{remark}
Assume that we allow infinite summation in (\ref{Tprzyblizanie})
with additional requirements that
\begin{eqnarray}
\label{Tszeregi} \sum_{p=1}^{\infty} |c_p^{-1}\varphi^p|^2 <
\infty,\qquad \sum_{p=1}^{\infty}
|c_p^{-1}\varphi^p(a_j^p+ib_j^p)|^2 < \infty,
\end{eqnarray}
for $j=1,\dots,P$. Then it is obvious that (\ref{Tprzyblizanie})
can be considered as a Fourier series, and every square integrable
mapping from $\C\T^n$ to $\C^m$ with square integrable
differentials is of this form. Thus the set of separable hermitian
2-forms which are integrally representable on a torus is the
closure of the set of the 2-forms $\rho_\Phi$ given by
(\ref{Togolny}), in the  topology given by the pseudonorms in
(\ref{Tszeregi}). In particular, we have the following result.
\end{remark}

\begin{proposition}\label{product-form}
A product form $\rho=\conj{\sigma}\oh \sigma$ with
$\sigma=\varphi\ot \psi\ne0$ is integrally representable on
$\C\T^n$ iff the coordinates $\psi_1,\dots,\psi_n$ of $\psi$ are
commensurable, i.e., there exist integers
$a_1,\dots,a_n,b_1,\dots,b_n$ and a complex number $w\in\C$ such
that $\psi_j=w(a_j + ib_j)$, for all $j$.
\end{proposition}
\begin{proof}
The "if" part follows from Lemma \ref{Tlemat} with $P=1$. To prove
the converse, assume that a product hermitian 2-form $\rho$ is
integrally representable. Then, by the above remark, it has a
representation (\ref{Togolny}) with infinite sum and the
additional requirements (\ref{Tszeregi}). Since $\rho$ has rank
one, there exist  a product form $\varphi\ot\psi \in
\C^m\ot(\C^n)^*$ and complex numbers $w_1, w_2, \ldots$ such that
\begin{eqnarray*}
\varphi^p\ot \psi^p=w_p\,\varphi\ot\psi, \quad p=1,2,\ldots.
\end{eqnarray*}
Then, if $w_p\ne0$, we can write $\varphi^p = u_p\varphi$,
$\psi^p=v_p\psi$, with nonzero complex numbers $u_p,v_p$. Since
$\psi^p =c_p^{-1} (a^p+ib^p)$, we get $\psi=w(a^p+ib^p)$, where
$w=v_p^{-1} c_p^{-1}$.
\end{proof}

\section{Integral representation condition}\label{konieczne}

As we have seen in Proposition \ref{product-form}, there are
product forms that can be integrally represented on the torus. For
the complex linear space this is not the case, as we will see in
the following theorem.

\begin{theorem}
\label{hermitowskie2} (a) No product hermitian 2-form can be
expressed in the integral form (\ref{przedstawienie}) over $\C^n$.

\noindent (b) Furthermore, if a hermitian 2-form $\rho$ is
representable in the integral form~(\ref{przedstawienie}) over
$\C^n$, then it satisfies the following condition.
\begin{itemize}
\item[(IRC)] If there exist $\tphi_0\in(\C^m)^*,
\tpsi_0\in\C^n\setminus \{0\}$ such that
$\rho(\tphi_0\otimes\tpsi_0,\tphi_0\otimes\tpsi_0) = 0$, then for
all $\tpsi\in\C^n$ we have
$\rho(\tphi_0\otimes\tpsi,\tphi_0\otimes\tpsi) = 0$.
\end{itemize}
\end{theorem}

\noindent We will call (IRC) \emph{the integral representation
condition} (necessary for integral representability of $\rho$).

 \begin{proof} \noindent
(a) It is enough to prove this statement for the hermitian 2-form
$\rho_0 = \conj{\sigma}\oh\sigma$ with $\sigma = \f_1\otimes\e_1$,
where $\f_1,\ldots,\f_m$ is the standard basis in $\C^m$ and
$\e_1,\ldots,\e_n$ is the dual basis to the standard basis in
$\C^n$. The general case reduces to this one by linear changes of
coordinates in $\C^m$ and $\C^n$.

We will show that there is no map $\Phi$ such that $\rho_\Phi =
\conj{(\f_1\otimes\e_1)}\oh(\f_1\otimes\e_1)$. Assume that there
is such a map. Then from the definition of $\rho_\Phi$ we have
\begin{eqnarray*}
\delta_{i1}\delta_{k1} & = & (\rho_\Phi)_{i1k1} =
\frac{1}{(2i)^n}\int_{\C^n}
\conj{(\cpar_1\Phi_i(z))} \cpar_1\Phi_k(z) \ d\conj{z}\wedge dz, \\
0 & = & (\rho_\Phi)_{ijkl}  =  \frac{1}{(2i)^n}\int_{\C^n}
\conj{(\cpar_j\Phi_i(z))} \cpar_l\Phi_k(z) \ d\conj{z}\wedge dz
 \qquad \forall_{j\neq 1 \ \text{or} \ l\neq 1}
\end{eqnarray*}
For $i=k$ and $j=l$, the above equations imply that for any
$i=1,\ldots,m$
\begin{eqnarray*}
\| \cpar_1\Phi_i \|_{L^2} & = & \delta_{1i},  \\
\| \cpar_j\Phi_i \|_{L^2} & = & 0,  \qquad\ \, j=2,\ldots,n,
\end{eqnarray*}
where $\| \cdot \|_{L^2}$ denotes the standard $L^2$ norm. From
these expressions we deduce that for $i\neq 1$
$$
\cpar_j\Phi_i = 0 \quad \text{a.e.,}\qquad j=1,\ldots,n.
$$
Now from a remark to Theorem 4.6.10 in \cite{kranz} it follows
that, for  $i\neq 1$, the maps $\Phi_i:\C^n\to\C$ are holomorphic
on $\C^n$,  and since we assume that they are square integrable,
we have
$$
\Phi_i = 0 \qquad \forall\ i\neq 1.
$$
For $i=1$ we have
$$
\cpar_j\Phi_1 = 0 \quad \text{a.e.}\qquad \forall\ j\neq 1.
$$
Again from the remark mentioned above we obtain that for all $z_1\in\C$
the function
$$
\Phi_{z_1} := \Phi_1(z_1,\cdot,\ldots,\cdot):\C^{n-1}\to\C
$$
is holomorphic on $\C^{n-1}$.
Trying to find any nontrivial square integrable function $\Phi_1$
to our problem, take any square integrable function $g:\C^n\to\C$
which fulfils compatibility conditions $\cpar_j g=0$ for all
$j\neq 1$. Then by Theorem 4.6.11 in \cite{kranz}, there exists
locally square integrable solution to the equations
\begin{eqnarray}
\label{diff}
\cpar_j\Phi_1 = \delta_{1j}g, \qquad j=1,\ldots,n.
\end{eqnarray}
If, additionally, for some $g$ the solution $\Phi_1$ is square
integrable then, by the Fubini theorem, $\Phi_{z_1}$ must be
square integrable function for almost all $z_1$. This means that
$$
\Phi_{z_1} = 0, \qquad \text{for almost all }z_1,
$$
as the null function is the only holomorphic square integrable
function on $\C^{n-1}$. In consequence,
$$
\| \Phi_1 \|_{L^2}  =  0
$$
which, together with $\Phi_i=0$ for $i\neq 1$, contradicts the
inequality $\rho_\Phi\neq 0$.

\medskip

\noindent (b) Just like in the proof of (a), we can assume that
$\tphi_0=\tf_1\in(\C^m)^*$ and $\tpsi_0=\tie_1\in\C^n$ , where
$\tf_1,\ldots,\tf_m$ is the dual basis to the standard basis in
$\C^m$ and $\tie_1,\ldots,\tie_n$ is the standard basis in $\C^n$
(the case of $\tphi=0$ is trivial). By the assumption, there
exists a square integrable map $\Phi:\C^n\rightarrow\C^m$, such
that $\rho = \rho_\Phi$, in particular,
\begin{eqnarray*}
0 = \rho(\tphi_0\otimes\tpsi_0,\tphi_0\otimes\tpsi_0) & = &
 \frac{1}{(2i)^n}\int_{\C^n} |\cpar_1\Phi_1(z)|^2 \ d\conj{z}\wedge dz.
\end{eqnarray*}
Thus $\cpar_1\Phi_1 = 0$ and, by arguments as above, we deduce
that $\Phi_1(z_1,z_2,\ldots,z_n)$ is a holomorphic function of
$z_1$, for almost all $(z_2,\ldots,z_n)\in\C^{n-1}$. A holomorphic
function in $L^2(\C)$ must be identically zero, thus $\Phi_1=0$.
Therefore, for all $\tpsi\in\C^n$ we have
\begin{eqnarray*}
\rho(\tphi_0\otimes\tpsi,\tphi_0\otimes\tpsi) & = &
 \sum_{ijkl} \delta_{1i}\conj{\tpsi}_j\delta_{1k}\tpsi_l \frac{1}{(2i)^n}\int_{\C^n}\conj{(\cpar_j\Phi_i(z))}\cpar_l\Phi_k(z) \ d\conj{z}\wedge dz \\
& = &
\sum_{jl} \conj{\tpsi}_j\tpsi_l \frac{1}{(2i)^n}\int_{\C^n}\conj{(\cpar_j\Phi_1(z))}\cpar_l\Phi_1(z) \ d\conj{z}\wedge dz = 0,
\end{eqnarray*}
which is our assertion.
\end{proof}

 We will now analyze the integral representation condition.
Consider a tensor product $H = K \otimes L$, where $K$ and $L$ are
vector spaces over $\C$ of finite dimensions $m$ and $n$,
respectively. Let $\rho : H \times H \rightarrow \C$ be a nonzero
separable hermitian 2-form. We denote by $\q{\rho} : H \to \R$ the
quadratic form associated with $\rho$, i.e.
$$
\q{\rho}(v) = \rho(v,v).
$$
Define two sets $\Ker \subset K$ and $\keR \subset L$ by
\begin{eqnarray*}
\Ker & := & \{\tphi\in K \ | \ \forall \tpsi \in L \quad  \q{\rho}(\tphi\ot\tpsi)=0 \},\\
\keR & := & \{\tpsi\in L \ | \ \forall \tphi \in L \quad  \q{\rho}(\tphi\ot\tpsi)=0 \}.
\end{eqnarray*}
Let
$$
\rho = \sum_{p=1}^{P} \conj{(\sigma^p)} \oh \sigma^p,
$$
where $P\geq 1$ and $\sigma^p \in H^*$ are linear functionals $
\sigma^p = \varphi^p \otimes \psi^p$, with $\varphi^p \in K^*$ and
$\psi^p \in L^*$. Then $\Ker$ is the intersection of the kernels
of $\varphi^p$'s, and $\keR$ is the intersection of the kernels of
$\psi^p$'s i.e.
$$
\Ker = \bigcap_p \ker{\varphi^p} \subset K \ , \qquad
\keR = \bigcap_p \ker{\psi^p} \subset L.
$$
Indeed, if $\tphi$ belongs to the intersection of the kernels of
all $\varphi^p$'s, then it is obvious that $\tphi\in \Ker$. Vice
versa, if there exists $p_0$ such that $\varphi^{p_0}(\tphi) \neq
0$, then $\q{\rho}(\tphi\ot\tpsi) > 0$ for all $\tpsi \notin \ker
\psi^{p_0}$, as all terms $\conj{\sigma^p} \oh \sigma^p$ defining
$\q{\rho}$ are nonnegative quadratic forms. The proof of the
second equality is analogous. From these equalities it follows
that $\Ker \subset K$ and $ \keR \subset L$ are linear subspaces
of nonzero codimensions (since $\rho\ne0$).

Let $\tphi = \tphi_1 + \tphi_k \in K$, where $\tphi_k\in\Ker$.
Then
$$
\q{\rho}(\tphi\otimes\tpsi) = \q{\rho}(\tphi_1\otimes\tpsi).
$$
Thus, for any vector $\cphi \in K/\Ker$ the quadratic form
$\q{\rho}_{\cphi} : L \to \R$ given by
$$
\q{\rho}_{\cphi}(\tpsi) = \q{\rho}(\tphi\otimes\tpsi)
$$
is well defined.

\begin{proposition}
\label{IRC}
For any separable hermitian 2-form $\rho$ the following conditions are equivalent.
\begin{enumerate}
    \item \label{pierwszy} (IRC) If there exist $\tphi_0\in K$ and $\tpsi_0\in L\backslash \{0\}$ such that $\q{\rho}(\tphi_0\otimes\tpsi_0) = 0$, then for all $\tpsi\in L$ we have $\q{\rho}(\tphi_0\otimes\tpsi) = 0$.
    \item \label{drugi} If $\tphi\in K$, $\tpsi\in L\setminus\{0\}$ and $\q{\rho}(\tphi\otimes\tpsi) = 0$ then $\tphi\in\Ker$.
    \item \label{trzeci} For all $0\neq\cphi \in K/\Ker$ the quadratic form $\q{\rho}_{\cphi}$ is strictly positive definite.
\end{enumerate}
\end{proposition}
\begin{proof}
(\ref{pierwszy} $\Rightarrow$ \ref{drugi}) This implication is obvious.
\\
(\ref{drugi} $\Rightarrow$ \ref{trzeci}) Assume that there exists $0\neq\cphi \in K/\Ker$ such that $\q{\rho}_{\cphi}$ is not strictly positive definite. Then there exist $\tpsi\in L\setminus\{0\}$ such that $\q{\rho}_{\cphi}(\tpsi)=0$. From the definition of $\q{\rho}_{\cphi}$ and condition \ref{drugi} we have that $\tphi\in \Ker$ and therefore $\cphi = 0$.
\\
(\ref{trzeci} $\Rightarrow$ \ref{pierwszy}) Assume that there exists $\tphi\in K, \tpsi\in L\backslash \{0\}$ such that $\q{\rho}(\tphi\otimes\tpsi) = 0$. Then condition \ref{trzeci} implies that $\cphi=0\in K/\Ker$ i.e. $\tphi\in\Ker$ and therefore for all $\tpsi\in L$ we have $\q{\rho}(\tphi\otimes\tpsi) = 0$.
\end{proof}

\begin{corollary}
If the integral representation condition holds for a nonzero
separable hermitian 2-form $\rho$ then $\keR = 0$. In particular,
if $\rho$ is integrally representable on $\C^n$, then $\keR = 0$.
\end{corollary}
\begin{proof}
Assume that $\keR \neq \{0\}$, and take any $0\neq\tphi \in K
\setminus \Ker$ and $\tpsi \in \keR\setminus\{0\}$. Then
$\q{\rho}(\tphi\otimes\tpsi) = 0$, which contradicts condition
\ref{drugi} in Proposition \ref{IRC}. The second statement follows
from Theorem \ref{hermitowskie2}.
\end{proof}

\begin{remark}
Note that the negation of (IRC) gives an almost sufficient
condition for entanglement of hermitian 2-forms. Namely, if a
positive definite hermitian 2-form does not satisfy (IRC), then it
belongs to the boundary of the cone of separable 2-forms, or it is
entangled, by Theorem \ref{hermitowskie} (b) and Theorem
\ref{hermitowskie2}.
\end{remark}

\begin{proposition}\label{rank}
If a nonzero hermitian 2-form $\rho$ is integrally representable
then it has rank at least $n=\dim{L}$.
\end{proposition}
\begin{proof}
As we have mentioned above, since $\rho\ne0$, codimension of
$\Ker$ in $K$ is nonzero. Thus $K/\Ker$ is nonempty. Therefore,
using Theorem \ref{hermitowskie2}(b) and condition \ref{trzeci} in
Proposition \ref{IRC}, we deduce that if $\rho$ is integrally
representable then there exists $\cphi \in K/\Ker$ such that the
quadratic form $\q{\rho}_{\cphi}$ on $L$ is strictly positive
definite i.e. has rank $n$. Thus  there exists $n$-dimensional
subspace in $K\otimes L$ such that the quadratic form $\q{\rho}$
is strictly positive when restricted to this subspace. Thus its
rank is at least $n$.
\end{proof}

One could expect that the rank of the integrally representable
hermitian 2-form $\rho$ always equals $\codim(\Ker)\cdot\dim(L)$.
An example presented below contradicts this assertion. The example
is a special case of the following proposition, which gives an
interesting class of integrally representable forms.

\begin{proposition}
Let $f:\C^n\to\C$ be a twice differentiable function (in the real
sense), such that its first and second conjugate derivatives are
square integrable. Denote by
$$
\Phi = \conj{D}f=(\conj{\partial}_1f,\dots,\conj{\partial}_nf) :
\C^n\to \C^n
$$
the conjugate gradient of $f$. Then the hermitian 2-form
$\rho_\Phi$ given by (\ref{przedstawienie}) is separable on
$\C^n\otimes(\C^n)^*$ and its kernel contains all antisymmetric
tensors.
\end{proposition}
\begin{proof}
From Theorem \ref{hermitowskie} the hermitian 2-form $\rho_\Phi$
is separable. So we need to check that $\q{\rho}_\Phi(c)=0$ for
every antisymmetric tensor $c$. But this is obvious since the
second conjugate derivative is a symmetric operator and therefore
\begin{eqnarray*}
\q{\rho}_\Phi(c) & = & \sum_{ijkl}\frac{1}{(2i)^n}\int_{\C^n} \conj{(c_{ij}\cpar_i\cpar_j f(z))}c_{kl}\cpar_k\cpar_l f(z) \ d\conj{z}\wedge dz \\
& = & \frac{1}{(2i)^n}\int_{\C^n}
\Big|\sum_{ij}c_{ij}\cpar_i\cpar_j f(z)\Big|^2 \ d\conj{z}\wedge
dz = 0, \quad \text{if} \ c_{ij}=-c_{ji}.
\end{eqnarray*}
\end{proof}

\begin{example}
Consider a function $f:\C^n\to\C$ given by
\begin{eqnarray*}
f(z) = f_\psi(z) = \Big(\frac{2}{\sqrt{\pi}\alpha}\Big)^n e^{-\frac{1}{2\alpha^2}(4|z|^2
-4\alpha^2\langle \psi,\conj{z}\rangle+\alpha^4|w|^2)},
\end{eqnarray*}
where $\alpha > 0$ and $\psi\in(\C^n)^*$. Denote by $\Phi(z) =
\conj{D}f(z)$ the conjugate differential of this function and
consider the mapping $\Phi:\C^n\to\C^n$. Thus $\Phi$ determines a
separable hermitian 2-form $\rho_\psi$ on $\C^n\ot(\C^n)^*$ with
coefficients
\begin{eqnarray*}
\lefteqn{(\rho_\psi)_{ijkl}  =  \frac{1}{(2i)^n}\int_{\C^n} \conj{(\cpar_i\cpar_j f(z))}\cpar_k\cpar_l f(z) \ d\conj{z}\wedge dz } \\
&& =  \frac{16}{(2i)^n}\int_{\C^n} \conj{(\psi_i - \frac{z_i}{\alpha^2})}\conj{(\psi_j - \frac{z_j}{\alpha^2})}(\psi_k - \frac{z_k}{\alpha^2})(\psi_l - \frac{z_l}{\alpha^2})|f(z)|^2 \ d\conj{z}\wedge dz.
\end{eqnarray*}
Clearly, if $\psi_s=\psi_s^r + i\psi_s^i$, for $s=1,\ldots,n$, then
\begin{eqnarray*}
|f(z)|^2 =  \Big(\frac{2}{\sqrt{\pi}\alpha}\Big)^{2n} \prod_{s=1}^n e^{-\frac{4}{\alpha^2}(x_s-\frac{\alpha^2\psi_s^r}{2})^2}
e^{ -\frac{4}{\alpha^2}(y_s-\frac{\alpha^2\psi_s^i}{2})^2}.
\end{eqnarray*}
Hence, using the Fubini theorem, one can integrate real and imaginary parts separately. Using standard expressions for first four gaussian moments, one can check that
\begin{eqnarray*}
(\rho_\psi)_{ijkl} & = &
\conj{\psi}_i\conj{\psi}_j\psi_k\psi_l \\
&& + \frac{1}{\alpha^2}(\conj{\psi}_i\psi_l\delta_{jk} + \conj{\psi}_j\psi_k\delta_{il} + \conj{\psi}_i\psi_k\delta_{jl} + \conj{\psi}_j\psi_l\delta_{ik}) \\
&& + \frac{1}{\alpha^4}(\delta_{ik}\delta_{jl} + \delta_{jk}\delta_{il}).
\end{eqnarray*}
Now if we evaluate $\q{\rho}_\psi$ at a general element $c\in\C^n\otimes(\C^n)^*$,
we obtain
\begin{eqnarray*}
\q{\rho}_\psi(c) & = & \sum_{ijkl}(\rho_\psi)_{ijkl}\conj{c}_{ij}c_{kl} \\
& = & \Big|\sum_{ij}\psi_i c_{ij}\psi_j\Big|^2 + \frac{1}{\alpha^2}\sum_j\Big|\sum_i\psi_i(c_{ij} + c_{ji})\Big|^2 + \frac{1}{2\alpha^4}\sum_{i,j}|c_{ij} + c_{ji}|^2.
\end{eqnarray*}
Therefore we see that $\q{\rho}_\psi(c) = 0$ for every antisymmetric tensor $c$.
On the other hand if $c$ is not antisymmetric then the third sum of the above expression
gives positive contribution to it. Thus rank of $\rho_\psi$ equals
$\frac{n(n+1)}{2}$ that is the codimension of the space of antisymmetric tensors.
\end{example}

\section{Concluding remarks}

We presented integral formulas for separable mixed states of
bi-partite finite dimensional systems. The states which can be
integrally represented are automatically separable. Almost all
separable states (in particular, all lying in the interior of the
set of such states) can be represented in the integral form.

There are natural questions related to our results.

\begin{itemize}
 \item[Q1] The map $\Phi$ in the integral formula for a given state
is not uniquely determined by the state. It would be advantageous
to isolate a subclass of maps $\Phi$ in which the representation
is unique. Does it exist such a subclass?
 \item[Q2] Can the results be generalized to $H=K\otimes L$,
with infinite dimensional $K$ or $L$?
 \item[Q3] Can they be generalized
to multi-partite systems?
\end{itemize}

We do not know the answer to question Q1. Answering question Q2 we
see that the space $L=\C^n$ can not be replaced by an infinite
dimensional Hilbert $\tilde L$ space because there is no natural
measure on $\tilde L$ to be used in the integral representation.
On the other hand, it is possible to replace $K=(\C^m)^*$ with an
infinite dimensional Hilbert space $\tilde K$. In this case a map
$\tilde \Phi:\C^n\to \tilde K^*$ should play the role of the
previous map $\Phi:\C^n\to\C^m$ and statements (a) in Theorems
\ref{hermitowskie} and \ref{hermitowskie-t} remain true, with
almost the same proofs. However, statements (b) can not be true as
Proposition \ref{separ} does not hold in the case of infinite
dimension. Namely, in this case the cone of separable hermitian
forms is closed, convex and nowhere dense in the space of all
hermitian operators (see \cite{clifton}). Statements weaker then
(b) follow from our results, by taking maps $\tilde \Phi$ with
images in finite dimensional subspaces of $\tilde K$.

Question Q3 seems to have a negative answer if we try to
generalize our approach literally. However, there is a different
way of representing separable states in an integral form, which
works for multipartite systems, too. This approach is a subject of
a forthcoming paper by the same authors.

\end{document}